\documentclass[a4paper,12pt]{article}

\usepackage[T1]{fontenc}
\usepackage[utf8]{inputenc}
\usepackage{graphicx}
\usepackage{xcolor}
\usepackage[english]{babel}

\usepackage{tgheros}
\usepackage{amsmath,amssymb,amsthm,textcomp}
\usepackage{enumerate}
\usepackage{multicol}

\usepackage{tikz}
\usepackage{tikz-cd}
\usetikzlibrary{babel,arrows,chains,matrix,positioning,scopes,calc,decorations,
decorations.pathreplacing}

\usepackage{geometry}
\newcommand{\linia}{\rule{\linewidth}{0.5pt}}
\newcommand{\set}[1]{ \{#1\} }

\newtheorem{theorem}{Theorem}

\makeatletter
\renewcommand{\maketitle}{
\begin{center}
\vspace{2ex}
{\huge \textsc{\@title}}
\vspace{1ex}
\\
\linia\\
\@author\\
\@date
\vspace{4ex}
\end{center}
}
\makeatother

\usepackage{fancyhdr}
\usepackage{listings}
\usepackage{authblk}

\begin{document}

\title{Path Optimization Sheaves}

\author[1]{Michael Moy}
\author[2]{Robert Cardona}
\author[2]{Robert Green}
\author[3]{Jacob Cleveland}
\author[4]{Alan Hylton}
\author[4]{Robert Short}
\affil[1]{Colorado State University}
\affil[2]{University at Albany - State University of New York}
\affil[3]{University of Nebraska at Omaha}
\affil[4]{NASA Glenn Research Center}


\date{\today}

\maketitle

\begin{abstract}

    Motivated by efforts to incorporate sheaves into networking, we seek to reinterpret pathfinding algorithms in terms of cellular sheaves, using Dijkstra's algorithm as an example.
    We construct sheaves on a graph with distinguished source and sink vertices, in which paths are represented by sections.
    The first sheaf is a very general construction that can be applied to other algorithms, while the second is created specifically to capture the decision making of Dijkstra's algorithm.
    In both cases, Dijkstra's algorithm can be described as a systematic process of extending local sections to global sections.
    We discuss the relationship between the two sheaves and summarize how other pathfinding algorithms can be interpreted in a similar way.
    While the sheaves presented here address paths and pathfinding algorithms, we suggest that future work could explore connections to other concepts from graph theory and other networking algorithms.
    
    This work was supported by the NASA Internship Project and SCaN Internship Project during the summer of 2020.    
    
    
    
\end{abstract}

\section{Introduction}

    An important goal of networking as a science is to study the effect of local decision making on global structure.  While graph theory has proved to be a useful mathematical model for understanding networking structures, the fact that graphs are traditionally static structures has led to concerns about their applicability to mobile wireless routing and delay-tolerant networking (DTN) problems.  In the search for a better model for understanding networks, sheaves have arisen as a more rigorous way to describe local actions impacting global structures.  Sheaves have proven useful in modeling wireless routing and modeling information flow through a network (\cite{ref_protocol},\cite{ref_flowcut}, \cite{ref_ncsheaf}). In addition, the High-rate Delay-Tolerant Networking (HDTN) project at NASA Glenn Research Center has been examining sheaf models of delay-tolerant networking as a framework for improving DTN functionalities (\cite{aeroconf2020}).
    
    Our goal in this paper is to translate Dijkstra's algorithm (\cite{ref_dijkstra}) into the language of sheaves as further evidence that sheaves are an effective model for networking.  As such, in this paper we define two sheaves: the "path sheaf," which generally models paths from a source to a sink, and the "distance path sheaf," which expands the path sheaf to account for weighted edges. These two sheaves are related by a simple sheaf morphism. For both sheaves, we interpret Dijkstra's algorithm as a process of extending local sections to global sections.  
    We also make a general claim that other pathfinding algorithms meeting certain requirements can be understood in terms of the path sheaf, and suggest that our construction of the distance path sheaf could be mimicked for other algorithms. 
    This paper concludes with a section containing suggestions for further work in this direction.  This includes ideas for describing other algorithms related to networking in terms of sheaves and for exploring further connections between sheaves and graph theory.
    
    
    This work was sponsored by the NASA Internship Program and SCaN Internship Program at NASA Glenn Research Center during the summer of 2020.  
    
    \subsection{Graph Theory Basics}

    For the purposes of this paper we will be considering finite graphs $G=( V,E)$ where $V$ is the set of vertices (also called nodes) and $E$ is the set of edges.  Each edge $e \in E$ is a set of two distinct elements of $V$, such as $e = \{v_1,v_2\}$. We define the degree $deg(v)$ of a vertex $v$ to be the number of edges that contain $v$, and will write $deg_G(v)$ if there is a need to specify the graph $G$. We will define a path as a sequence of edges $(e_0, e_1, \dots,e_{n-1})$ connecting a sequence of distinct vertices $(v_0, v_1, \dots, v_n)$ called the vertex sequence, such that $e_i = \{v_i, v_{i+1}\}$ for each $i$. Similarly, we will define a cycle as a sequence of edges $(e_0, e_1, \dots,e_{n-1})$ connecting a vertex sequence $(v_0, v_1, \dots, v_n)$, such that $e_i = \{v_i, v_{i+1}\}$ for each $i$ and $v_0 = v_n$, with all other vertices distinct. In some instances in this paper we consider a weighted graph, which we define as a graph $G$ with a map $w:E \rightarrow \mathbb{R}^+$ where $w(e)$ is called the weight of $e$.

   \subsection{Sheaf Theory Basics}
   
   Because we are working over graphs, we will use the theory of cellular sheaves to define the sheaves in this paper.  While we summarize the key definitions below, a more formal treatment can be found in \cite{ref_curry}.  For a more classical treatment of sheaves, we recommend \cite{ref_bredonshf} or \cite{ref_swan}.
   
   A sheaf over a graph $G=(V,E)$ is a functor $\mathcal{F}$ assigning to each vertex and edge of $G$ a set $\mathcal{F}(v)$ (or $\mathcal{F}(e)$ respectively).  In addition, for each vertex $v$ and edge $e$ such that $v \in e$, then there is a restriction map $\mathcal{F}(v\rightsquigarrow e): \mathcal{F}(v) \to \mathcal{F}(e)$.  
    A (global) section $s$ is a choice of elements in each set over each vertex and edge of $G$ such that for any vertex $v$ and edge $e$ such that $v \in e$, $\mathcal{F}(v\rightsquigarrow e)(s(v))=s(e)$.  If $H \subseteq V \cup E$, then a local section $s$ (over $H$) is a section defined only over elements in $H$.
   
   A good way to think about sheaves is as a way of organizing and relating data moving across a network, such as in \cite{aeroconf2020}.  For a sheaf $\mathcal{F}$ over a graph $G$, the sets over the vertices can represent what information can be stored in each node while the edges can represent what information can be transmitted over each edge.  Restriction maps then prescribe the relationship between data transmitted over edges and the data stored in the vertices.  A section prescribes information to each vertex and edge in the graph that is consistent with the relationships we expect to see in the data.  The real power of sheaves comes from defining restriction maps well so that the sections provide actual solutions to different problems.

    In this paper, the information being organized and related using sheaves is not the actual information moving through a network, but instead represents capacities or amounts of data that can be moved.  That being said, it is still useful to think of sheaves as layering information above a network, but the information in our case will be information about limits of resources rather than information about the resources themselves.    
    
    \subsection{Acknowledgements}
    This work was supported by the NASA Internship Project and SCaN Internship Project during the summer of 2020.  We would like to thank Gabriel Bainbridge, Joseph Ozbolt, Michael Robinson, Justin Curry, and John Nowakowski for many helpful discussions and for providing advice and guidance while writing this paper.
   
\section{Path Sheaf}
    For our first sheaf, we assume that we have a finite graph $G=(V,E)$ with two distinguished vertices: a source vertex $v_S$ and sink vertex $v_T$.  In addition, we require that $\deg(v) \geq 2\hspace{0.1cm}$ for all $v\in V$, except the sink $v_T$ and the source $v_S$, which we will require only to have degree $\geq 1$. Define $E(v)$ to be the set of all edges in $G$ connected to $v\in V$ and let $H(v)= \{\{e_i,e_j\}\subseteq E(v)\mid e_i \neq e_j\}$ be the set of two-element subsets of $E(v)$.  In addition, let $H_o(v) = \{(e_i,e_j) \in E(v)^2 \mid e_i \neq e_j\}$ be the possible orderings of the sets in $H(v)$, which will be used later.  We will generally use the symbol $\top$ to represent an active edge, one that is being routed through, while $\bot$ represents an inactive object, whether a vertex or an edge.
    
    We define the path sheaf $\mathcal{P}$ on $G$ to map vertices by the rule
    \[
  \mathcal{P}(v) =
  \begin{cases}
                                    E(v) & \text{if $v=v_T$ or $v=v_S$} \\
                                   H(v)\cup\{\bot\} & \text{otherwise} 
  \end{cases},
\]

\noindent and edges by the rule
\[
\mathcal{P}(e) = \{\bot, \top\}.
\]
\noindent If $v$ is a sink or a source, define the restriction map
\[
\mathcal{P}(v\rightsquigarrow e)(e_i) = 
\begin{cases}
            \top & \text{if $e = e_i$} \\
                                   \bot & \text{otherwise} 
\end{cases}.
\]

\noindent If $v$ is a non-sink and a non-source node with assignment $\{e_i,e_j\}$, define the restriction map
\[
\mathcal{P}(v\rightsquigarrow e)(\{e_i,e_j\}) = 
\begin{cases}
    \top & \text{if $e = e_i$ or $e = e_j$}\\
    \bot & \text{otherwise}
\end{cases}.
\]

\noindent Lastly, if $v$ is as above except with assignment $\bot$, define the restriction map
\[
\mathcal{P}(v\rightsquigarrow e)(\bot) = \bot .
\]

\subsection{Sections of $\mathcal{P}$}

Because of how we constructed our sheaf, a section of $\mathcal{P}$ describes a collection of edges and nodes that are active. A section can be thought of as representing a path in the graph, where edges assigned to $\top$ are included in the path.  By the definition of the restriction maps, the value of a section at an active vertex is the set of two active edges (or the single active edge for the source or sink) connecting to the vertex.

Considering the following example given in Figure \ref{ex:arc_path}.  Tracking the active vertices and edges yields a path from $v_S$ to $v_T$.  In particular, $s_1$ carves out the path $(e_2,e_4,e_3)$.  We can see this by following the section values and restriction maps of the sheaf as a means of moving along the graph.  Here, $s_1(v_S) = e_2$, so $\mathcal{P}(v_S \rightsquigarrow e_2)(e_2) = \top = s_1(e_2)$ and $\mathcal{P}(v_S \rightsquigarrow e_1)(e_2) = \bot= s_1(e_1)$.  Since $s_1(e_2)=\top$, we then look to $s_1(v_3)$.  Since $s_1(v_3) = \set{e_2,e_4}$, we know both where we came from -- since $\mathcal{P}(v_3 \rightsquigarrow e_2)(\set{e_2,e_4}) = \top = s_1(e_2)$ -- and where we are going -- since $\mathcal{P}(v_3 \rightsquigarrow e_4)(\set{e_2,e_4})=\top$ and $\mathcal{P}(v_3 \rightsquigarrow e_5)(\set{e_2,e_4})=\bot$ which both correspond to the section values.  In this way, the section values at the vertices in some sense determine the section values at the edges via the restriction maps.  The reasons this is a section at all are that the restriction maps are respected and the values are consistent across the full sheaf.

\begin{figure}
\begin{tikzpicture}
          \draw [black, fill = black] (-4,0) circle [radius=0.1];
          \draw [black, fill = black] (4,0) circle [radius=0.1]; 
          \draw [black, fill = black] (0,0) circle [radius=0.1]; 
          \draw [black, fill = black] (0,3) circle [radius=0.1]; 
          \draw [black, fill = black] (4,3) circle [radius=0.1]; 
          \draw [black, fill = black] (-4,3) circle [radius=0.1]; 
          \draw [black] (-4,0) -- (0,0); 
          \draw [black] (4,0) -- (0,0); 
          \draw [black] (-4,3) -- (0,3); 
          \draw [black] (4,3) -- (0,3); 
          \draw [black] (-4,0) -- (-4,3); 
          \draw [black] (0,0) -- (0,3); 
          \draw [black] (4,0) -- (4,3); 
          
          \node at (6,1.5) {Graph $G$};
          
          \node at (-4.25,1.5) {$e_1$};
          \node at (-2,3.25) {$e_2$};
          \node at (-2,-0.25) {$e_3$};
          \node at (0.25,1.5) {$e_4$};
          \node at (2,-0.25) {$e_6$};
          \node at (2,3.25) {$e_5$};
          \node at (4.25,1.5) {$e_7$};
          \node at (-4.3,3.3) {$v_S$};
          \node at (-4.3,-0.3) {$v_T$};
          \node at (0,3.3) {$v_3$};
          \node at (0,-0.3) {$v_4$};
          \node at (4.3,3.3) {$v_5$};
          \node at (4.3,-0.3) {$v_6$};
\end{tikzpicture}

\begin{tikzpicture}
          \draw [black,->] (0,4.5) -- (0,3.5);
          \node at (0.3,4) {$s_1$}; 
          \draw [red] (-4,0) -- (0,0); 
          \draw [black,dashed] (4,0) -- (0,0); 
          \draw [red] (-4,3) -- (0,3); 
          \draw [black,dashed] (4,3) -- (0,3); 
          \draw [black,dashed] (-4,0) -- (-4,3); 
          \draw [red] (0,0) -- (0,3); 
          \draw [black,dashed] (4,0) -- (4,3); 
          \draw [red, fill = red] (-4,0) circle [radius=0.1];
          \draw [black] (4,0) circle [radius=0.1]; 
          \draw [red, fill = red] (0,0) circle [radius=0.1]; 
          \draw [red, fill = red] (0,3) circle [radius=0.1]; 
          \draw [black] (4,3) circle [radius=0.1]; 
          \draw [red, fill = red] (-4,3) circle [radius=0.1]; 
          
          \node at (8,1.5) {\color{red} Section $s_1$ values over $G$};
          \node at (-4.25,1.5) {$\bot$};
          \node at (-2,3.25) {$\top$};
          \node at (-2,-0.25) {$\top$};
          \node at (0.25,1.5) {$\top$};
          \node at (2,-0.25) {$\bot$};
          \node at (2,3.25) {$\bot$};
          \node at (4.25,1.5) {$\bot$};
          \node at (-4.3,3.3) {$e_2$};
          \node at (-4.3,-0.3) {$e_3$};
          \node at (0,3.3) {$\{e_2, e_4\}$};
          \node at (0,-0.3) {$\{e_3, e_4 \}$};
          \node at (4.3,3.3) {$\bot$};
          \node at (4.3,-0.3) {$\bot$};
\end{tikzpicture}
\caption{A graph $G$ in black with the section $s_1$ of $\mathcal{P}$ over $G$ in red.\label{ex:arc_path}}
\end{figure}
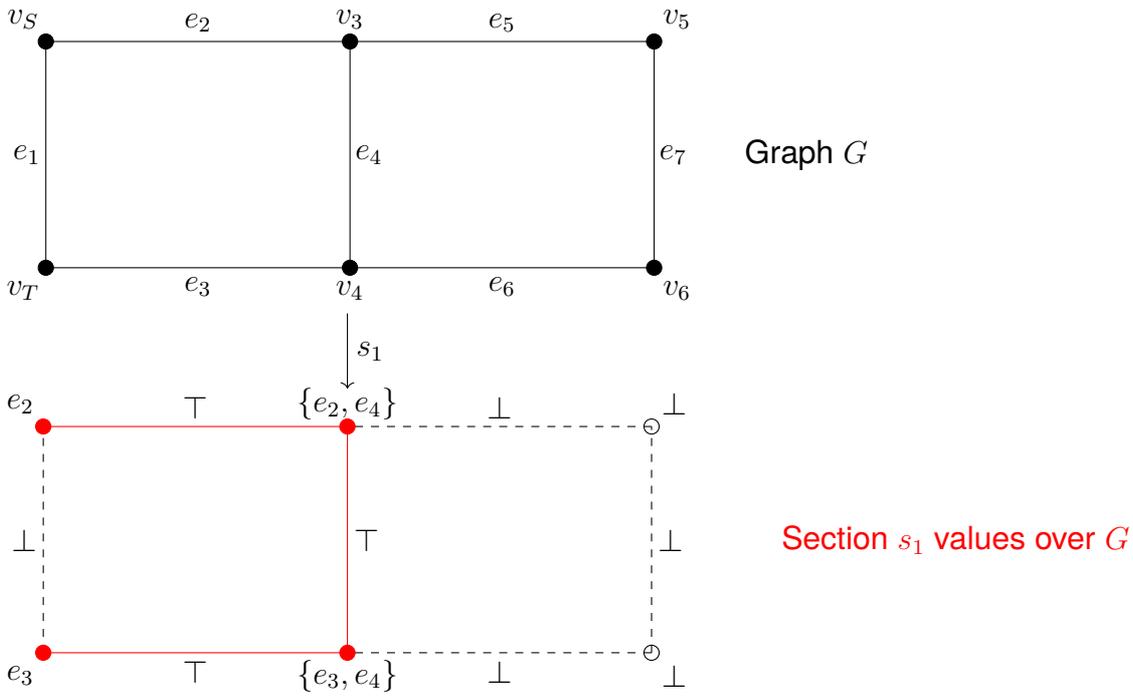

The fact that the section maps out a path from $v_S$ to $v_T$ is not unique to this section.  Indeed, it is a feature of sections of $\mathcal{P}$ generally.

\begin{theorem}\label{thm:SecToPath}
If $s$ is a global section of $\mathcal{P}$, then there exists a path $(e_0,\dots,e_n)$ from $v_S$ to $v_T$ such that $s(e_i) = \top$ for all $i$.  

\end{theorem}

\begin{proof}
Since $s$ must be consistent with the restriction maps of $\mathcal{P}$, for any edge $e$ connected to a vertex $v$, if $s(e) = \top$, then $s(v) \neq \bot$.  So there is a subgraph $G_s$ consisting of all vertices $v$ such that $s(v) \neq \bot$, and all edges $e$ such that $s(e) = \top$.  The restriction maps further imply that $v_S \in G_s$ and $v_T \in G_s$ with $deg_{G_s}(v_S) = 1$ and $deg_{G_s}(v_T) = 1$ and that the degrees of all other vertices in $G_s$ are 2.  Defining $e_0 = s(v_S)$, we have $e_0 \in G_s$.

Since $G_s$ is finite, there exists a path beginning at $v_S$ that is of maximal length: let this path be formed by the edges $(e_0, e_1, \dots,e_{n-1})$ with vertex sequence $(v_S,v_1, \dots, v_n)$.  If $deg_{G_s}(v_n) = 2$, then there is an edge $e_n = \{v_n, v_{n+1}\}$ for some $ v_{n+1} \neq v_{n-1}$.  We cannot have $v_{n+1} = v_S$, as this would imply $deg_{G_s}(v_S) \geq 2$, and we cannot have $v_{n+1} = v_i$ for any $i<n-1$, as this would imply $deg_{G_s}(v_i) \geq 3$.  Therefore $v_{n+1}$ is distinct from $v_S, v_1, \dots, v_n$, which means $(e_0, e_1 \dots,e_n)$ is a path.  This, however, contradicts our choice that $(e_0, e_1 \dots,e_{n-1})$ has maximal length.  Therefore $deg_{G_s}(v_n) = 1$, and since $v_n \neq v_S$, this implies $v_n = v_T$.  Thus, $(e_0, e_1 \dots,e_{n-1})$ is a path from $v_S$ to $v_T$ and $s(e_i) = \top$ for all $i$.
\end{proof}


Paths in $G$ and sections of $\mathcal{P}$ are further related by the following theorem.

\begin{theorem} \label{thm:PathToSec}
Let $(e_0, e_1, \dots,e_{n-1})$ be a path in $G$ with vertex sequence $(v_0, v_1, \dots, v_n)$ such that $v_i$ is neither the source nor the sink for $1 \leq i \leq n-1$. Then there exists a local section $s$ of $\mathcal{P}$ defined on the edges and vertices of the path, with $s(e_i) = \top$ for all $i$.  If $v_0$ is the source and $v_n$ the sink or vice versa, then this $s$ can be extended to a global section.  Similarly, for any cycle that does not pass through the source or sink, there exists a local section $s$ such that $s(e) = \top$ for all edges $e$ in the cycle.
\end{theorem}

\begin{proof}
For a path $(e_0, e_1, \dots,e_{n-1})$ as described, define $s$ as follows: 

$$s(e_i) = \top \text{ for all } 0 \leq i \leq n-1$$ 
$$s(v_i) = \{e_{i-1}, e_i\} \text{ for all } 1 \leq i \leq n-1$$

$$
s(v_0) = 
\begin{cases}
    e_0 & \text{if $v_0 = v_S$ or $v_0 = v_T$}\\
    \{e_0, e'\} & \text{otherwise, for any $e' \neq e_0$ connected to $v_0$}
\end{cases}
$$

$$
s(v_n) = 
\begin{cases}
    e_{n-1} & \text{if $v_n = v_S$ or $v_n = v_T$}\\
    \{e_{n-1}, e''\} & \text{otherwise, for any $e'' \neq e_{n-1}$ connected to $v_n$}
\end{cases}
$$

Note that choices of $e'$ and $e''$ exist by the assumption that nodes other than the source or sink have degree at least 2.  It can be checked that $s$ is consistent with the restriction maps of $\mathcal{P}$ on $v_0, v_1, \dots, v_n$ and $e_0, e_1, \dots,e_{n-1}$, so $s$ is a local section.  If $v_0$ is the source and $v_n$ the sink or vice versa, then $s$ can be extended to a global section by setting $s(v) = \bot$ for all vertices $v$ not in the path and $s(e) = \bot$ for all edges not in the path.  Again, it can be checked that the extended $s$ is consistent with the restriction maps.  The case of cycles that do not pass through the source or sink can be handled in a similar way.
\end{proof}

While Theorem \ref{thm:SecToPath} shows that all global sections have an active path from source to sink, it is possible for a global section to have additional cycles that are disjoint from this path. Consider the section $s_2$ over $G$ depicted in Figure \ref{ex:section+extra_loop}.  Notice that $s_2$ has both a path from $v_S$ to $v_T$ -- given by $(e_1)$ -- as well as a cycle between $v_3$, $v_4$, $v_5$, and $v_6$ -- given by $(e_4,e_6,e_7,e_5)$ for example. Sections like this will not be relevant to pathfinding algorithms, so we will be careful to describe the algorithms below so that the sections produced do not contain such cycles.

\begin{figure}\label{ex:section+extra_loop}
\begin{tikzpicture}
          \draw [black, fill = black] (-4,0) circle [radius=0.1];
          \draw [black, fill = black] (4,0) circle [radius=0.1]; 
          \draw [black, fill = black] (0,0) circle [radius=0.1]; 
          \draw [black, fill = black] (0,3) circle [radius=0.1]; 
          \draw [black, fill = black] (4,3) circle [radius=0.1]; 
          \draw [black, fill = black] (-4,3) circle [radius=0.1]; 
          \draw [black] (-4,0) -- (0,0); 
          \draw [black] (4,0) -- (0,0); 
          \draw [black] (-4,3) -- (0,3); 
          \draw [black] (4,3) -- (0,3); 
          \draw [black] (-4,0) -- (-4,3); 
          \draw [black] (0,0) -- (0,3); 
          \draw [black] (4,0) -- (4,3); 
          
          \node at (6,1.5) {Graph $G$};
          
          \node at (-4.25,1.5) {$e_1$};
          \node at (-2,3.25) {$e_2$};
          \node at (-2,-0.25) {$e_3$};
          \node at (0.25,1.5) {$e_4$};
          \node at (2,-0.25) {$e_6$};
          \node at (2,3.25) {$e_5$};
          \node at (4.25,1.5) {$e_7$};
          \node at (-4.3,3.3) {$v_S$};
          \node at (-4.3,-0.3) {$v_T$};
          \node at (0,3.3) {$v_3$};
          \node at (0,-0.3) {$v_4$};
          \node at (4.3,3.3) {$v_5$};
          \node at (4.3,-0.3) {$v_6$};
\end{tikzpicture}

\begin{tikzpicture}
          \draw [black,->] (0,4.5) -- (0,3.5);
          \node at (0.3,4) {$s_2$}; 
          \draw [black,dashed] (-4,0) -- (0,0); 
          \draw [red] (4,0) -- (0,0); 
          \draw [black,dashed] (-4,3) -- (0,3); 
          \draw [red] (4,3) -- (0,3); 
          \draw [red] (-4,0) -- (-4,3); 
          \draw [red] (0,0) -- (0,3); 
          \draw [red] (4,0) -- (4,3); 
          \draw [red, fill = red] (-4,0) circle [radius=0.1];
          \draw [red, fill = red] (4,0) circle [radius=0.1]; 
          \draw [red, fill = red] (0,0) circle [radius=0.1]; 
          \draw [red, fill = red] (0,3) circle [radius=0.1]; 
          \draw [red, fill = red] (4,3) circle [radius=0.1]; 
          \draw [red, fill = red] (-4,3) circle [radius=0.1]; 
          
          \node at (8,1.5) {\color{red} Section $s_2$ values over $G$};
          \node at (-4.25,1.5) {$\top$};
          \node at (-2,3.25) {$\bot$};
          \node at (-2,-0.25) {$\bot$};
          \node at (0.25,1.5) {$\top$};
          \node at (2,-0.25) {$\top$};
          \node at (2,3.25) {$\top$};
          \node at (4.25,1.5) {$\top$};
          \node at (-4.3,3.3) {$e_1$};
          \node at (-4.3,-0.3) {$e_1$};
          \node at (0,3.3) {$\{e_4, e_5\}$};
          \node at (0,-0.3) {$\{e_4, e_6 \}$};
          \node at (4.3,3.3) {$\{e_5, e_7\}$};
          \node at (4.3,-0.3) {$\{e_6, e_7\}$};
\end{tikzpicture}
\caption{A graph $G$ and the values of a particular global section $s$}
\end{figure}
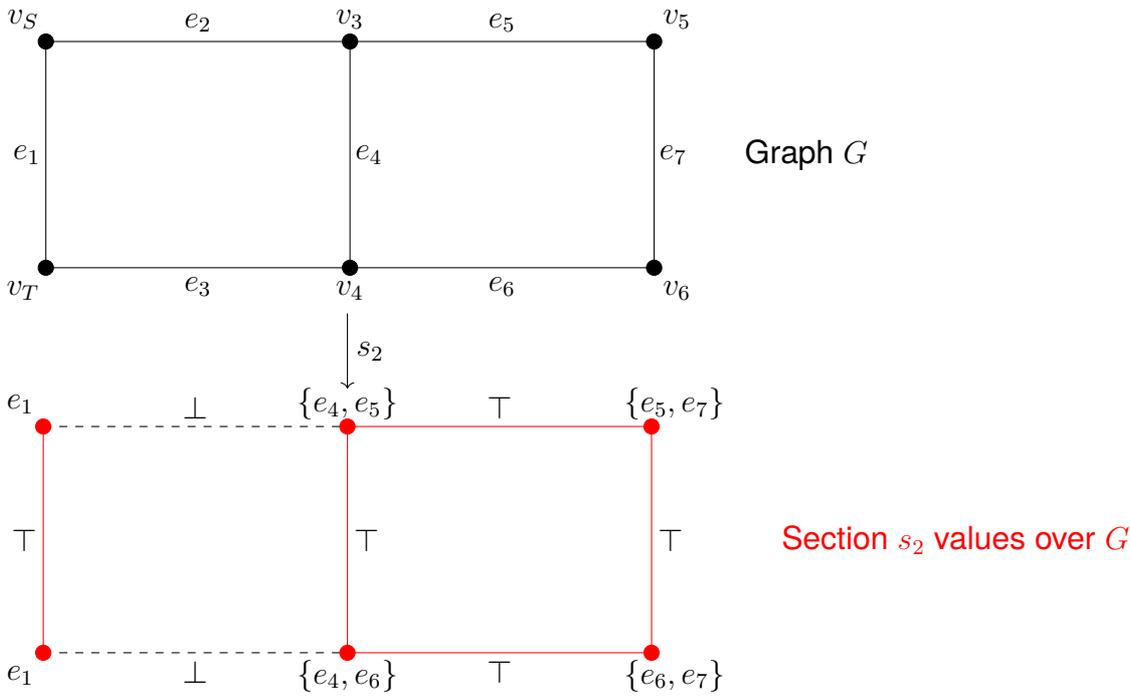



Note also that we have not required that the graph $G$ be connected.  In the case that $G$ does not have a path from source to sink, Theorem \ref{thm:SecToPath} implies that $\mathcal{P}$ does not have any global sections.

We'll now move toward discussing pathfinding algorithms in terms of $\mathcal{P}$.  To this end, suppose $\mathcal{P}$ is built on a weighted graph $G = (V,E)$ with weight function $w:E \rightarrow \mathbb{R}^{+}$. Define a cost function $c$ from the space of local sections of $\mathcal{P}$ to $\mathbb{R}^+$ by

\[
c(s) = \sum_{s(e) = \top} w(e)
\]

\noindent where the sum is taken over all edges $e$ in the domain of $s$ that are activated.  This cost function represents the cost of traveling along all active edges in a section.


\subsection{Dijkstra's Algorithm in Terms of the Path Sheaf}

The relationships prescribed by Theorems \ref{thm:SecToPath} and \ref{thm:PathToSec} between sections of $\mathcal{P}$ and paths in $G$ allow us to describe routing algorithms using $\mathcal{P}$. An algorithm that searches through paths in order to find an optimal path from source to sink can be viewed as searching through the sections of $\mathcal{P}$ in order to find a particular global section. We'll use Dijkstra's algorithm as an example of a routing algorithm that can be viewed in terms of $\mathcal{P}$, using the cost function described above to facilitate the decision making of the algorithm.

Dijkstra's algorithm systematically searches through paths beginning at the source vertex in a weighted graph. It records at each step a tentative distance to each vertex, which is the sum of edge weights along a specific path from the source to that vertex. Each path considered in this process corresponds to a local section $s$ of $\mathcal{P}$, as described in Theorem \ref{thm:PathToSec}, and the distance to the vertex at the end of this path is given by $c(s)$. Thus, we can view this process as recording a tentative section for each vertex, where the section represents the current best known path to that vertex and its cost corresponds to the tentative distance stored by Dijkstra's algorithm. Each new path considered in Dijkstra's algorithm extends a previous path by one edge: this translates to extending a tentative section to a new section defined on an additional edge and vertex. The order in which new paths are considered is determined by the tentative distances, which in our interpretation means they are determined by the costs of the tentative sections. Once paths are found that reach the sink, their corresponding sections can be extended to global sections by Theorem \ref{thm:PathToSec}. The algorithm terminates once it can be confirmed that the shortest path from source to sink has been found: this corresponds to the global section with minimal cost, as is implied by Theorem \ref{thm:SecToPath}. Also by Theorem \ref{thm:SecToPath}, if no path from source to sink exists, then no global section exists. In this case, the algorithm terminates when it can be confirmed that no path from source to sink exists.

We have thus reinterpreted Dijkstra's algorithm as a search over the local sections of $\mathcal{P}$ that gradually extends sections from local to global. The process of searching through these sections is governed by the cost function, and the end result is the global section with minimal cost. We will discuss the possibility of describing other algorithms in terms of $\mathcal{P}$ later, as well as provide an explicit description of Dijkstra's algorithm in terms of a closely related sheaf, defined in the following section.

\section{Distance Path Sheaf}



While $\mathcal{P}$ encapsulates the idea of paths from source to sink in a general way, describing Dijkstra's algorithm in terms of $\mathcal{P}$ required the use of the cost function.  By incorporating this cost function directly into a sheaf, we obtain a new sheaf specially suited to describe Dijkstra's algorithm.  In this new sheaf, which we will call the distance path sheaf, $\mathcal{DP}$, the distance from the source is recorded for active edges and vertices.

We will work with a weighted graph $G=(V,E)$ meeting the same requirements on degrees as for $\mathcal{P}$, with weight function $w:E \rightarrow \mathbb{R}^{+}$.  Then $\mathcal{DP}$ maps vertices 
\[
\mathcal{DP}(v)=
\begin{cases}
E(v)\times \{0\} &\text{if $v=v_S$}\\
E(v)\times \mathbb{R}^+ &\text{if $v=v_T$}\\
(H_o(v)\times \mathbb{R}^+)\cup \{\bot\} &\text{otherwise}
\end{cases}
\]

\noindent and edges

\[
\mathcal{DP}(e)=\mathbb{R}^+ \cup \{\bot\}
\]

\noindent accordingly. Note that we use $H_o(v)$ here instead of $H(v)$ as the order will matter significantly in defining the restriction maps.  For the source, define the restriction map

\[
\mathcal{DP}(v_S\rightsquigarrow e)(e_i,0) = 
\begin{cases}
w(e) & \text{if $e = e_i$}\\
\bot & \text{otherwise}
\end{cases}.
\]

\noindent But, for the sink, define the restriction map

\[
\mathcal{DP}(v_T\rightsquigarrow e)(e_i,x) = 
\begin{cases}
x &\text{if $e=e_i$}\\
\bot &\text{otherwise}
\end{cases}.
\]

\noindent If $v$ is a non-sink and a non-source node with non-$\bot$ assignment, define

\[
\mathcal{DP}(v\rightsquigarrow e)(e_i,e_j,x)=
\begin{cases}
x &\text{if $e=e_i$}\\
x + w(e) &\text{if $e=e_j$}\\
\bot &\text{otherwise}
\end{cases}.
\]

\noindent Then, finally, if $\mathcal{DP}(v)=\bot$, define

\[
\mathcal{DP}(v\rightsquigarrow e)(\bot) = \bot.
\]

The pair of edges in $H_o(v)$ assigned by a section represent the incoming and outgoing edges.  In this way, sections are able to represent directed paths.  In a section,the restriction maps require that an edge activation should match the "upstream" vertex, and the vertex activation should contribute along with the edge weight to the "upstream" edge. This sheaf is particularly suited to Dijkstra's algorithm: a section representing a path starting at the source assigns numerical values to the vertices, representing their distances from the source.


Note that, unlike $\mathcal{P}$, loops are not possible in the global sections of $\mathcal{DP}$. Since weights are strictly positive, the weights of the edges along a loop would be accumulated but won't match the `starting' node for the loop.  This would fail to respect the restriction maps from the `initial' vertex.

\subsection{Dijkstra's Algorithm in Terms of the Distance Path Sheaf}

The following is a rewriting of Dijkstra's algorithm using $\mathcal{DP}$.  This is not meant as a computational alternative, but rather as a concrete example of how such an algorithm can be reinterpreted in the language of sheaves.  A notable feature of Dijkstra's algorithm is that the tentative distance to each node is associated to a tentative path to that node.  This has a natural interpretation based on the sections of $\mathcal{DP}$, since the sections store both a path and the distance to each node in that path. For convenience, to compare the distances stored by sections $s$ and $s'$, we will write $s(v) < s'(v')$ to mean that the last component of $s(v)$ is less than the last component of $s'(v')$ (these are of course the numerical components).

Dijkstra's algorithm can be described in terms of $\mathcal{DP}$ as follows:

\vspace{.5cm}


1. Mark all vertices as unvisited.  For each vertex $v$, the following steps will update a tentative section $s_v$.  For non-source vertices $v$, leave $s_v$ undefined initially, and for $v_S$, let $s_{v_S}$ be defined only on $v_s$, with $s_{v_S}(v_S) = (e, 0)$, where $e$ can be chosen as the first edge considered in the following step. Set the current vertex as $v_c = v_S$, and set the current section as $s_c = s_{v_S}$.

2. Look at the possible ways to extend $s_c$ along one edge: for each unvisited vertex $u$ adjacent to $v_c$, let $e = \{v_c, u\}$. If necessary, change the "outgoing edge" component of $s_c(v_c)$ to $e$\footnote{ This is the first component for the source and the second component for some vertex that is not the source.  Note that the sink will never be the current vertex.} and let $l$ be the numerical component of $s_c(v_c)$.  If $u \neq v_T$, then since $deg_G(u) \ge 2$, $u$ is connected to some other edge $f \neq e$, so define a section $s_{new}$ as an extension of $s_c$ to include $u$ and $e$ by $s_{new}(e) = l + w(e)$ and $s_{new}(u) = (e, f, l + w(e))$.  If $u = v_T$, define $s_{new}$ as an extension of $s_c$ to include $u$ and $e$ by $s_{new}(e) = l + w(e)$ and $s_{new}(u) = (e, l + w(e))$.  Compare the new section and the previously stored section for the vertex $u$: 
if $s_u$ has not yet been defined or if $s_{new}(u) < s_u(u)$, then set $s_u$ equal to $s_{new}$.

3. Once all unvisited vertices adjacent to the current vertex have been checked, mark the current node as visited.

4. If $s_{v}$ is undefined for all unvisited $v$, then the algorithm stops and no path from source to sink exists.  If $s_{v_T}(v_T) \leq s_{v}(v)$ for all unvisited $v$ such that $s_v$ has been defined, then the algorithm stops and $s_{v_T}$ represents the shortest path from source to sink.  In this case $s_{v_T}$ may be extended to a global section by assigning a value of $\bot$ to all vertices and edges not in its current domain, and the length of the path is given by the numerical component of $s_{v_T}(v_T)$.  Otherwise, choose the unvisited vertex $v$ with the least $s_v(v)$, set $s_c$ equal to $s_v$, set $v_c$ equal to $v$, and return to step 2.



\vspace{.5cm}

Rewriting the algorithm in this way involves more variables and more details in each step than a normal description of Dijkstra's, as the sections of the sheaf contain large amounts of information.  However, the way in which data is stored in sections of the sheaf reflects the structures inherent in Dijkstra's algorithm.  The tentative distances in Dijkstra's algorithm are always linked to tentative paths; both are described by tentative sections of $\mathcal{DP}$.  The process of extending tentative paths to entire paths from source to sink is captured by the local-to-global nature of the sheaf.  And finally, distance from the source, the primary concern of Dijkstra's algorithm, is encapsulated in the sections constructed.

\section{Relationships between Sheaves}

\subsection{A Sheaf Morphism}
Having discussed Dijkstra's algorithm in terms of both $\mathcal{P}$ and $\mathcal{DP}$, we now observe the relationship of $\mathcal{P}$ and $\mathcal{DP}$. There exists a sheaf morphism $\varphi: \mathcal{DP} \rightarrow \mathcal{P}$ that forgets the numerical values included in $\mathcal{DP}$.  Viewing $\mathcal{P}$ and $\mathcal{DP}$ as functors into \textbf{Set}, a sheaf morphism is a natural transformation.  
Specifically, we must define functions $\varphi_v: \mathcal{DP}(v) \rightarrow \mathcal{P}(v)$ for all vertices $v$ and $\varphi_e: \mathcal{DP}(e) \rightarrow \mathcal{P}(e)$ for all edges $e$ such that whenever $v \in e$, the following diagram commutes:

\begin{tikzpicture}

          \draw [->] (1,0) -- (3.3,0); 
          \draw [->] (1,3.5) -- (3.3,3.5); 
          \draw [->] (0,3) -- (0,.5); 
          \draw [->] (4,3) -- (4,.5);
          
          \node at (2.15,-.3) {$\varphi_e$};
          \node at (2.15,3.8) {$\varphi_v$};
          \node at (-1.25,1.75) {$\mathcal{DP}(v \rightsquigarrow e)$};
          \node at (5,1.75) {$\mathcal{P}(v \rightsquigarrow e)$};
          
          \node at (0,0) {$\mathcal{DP}(e)$};
          \node at (0,3.5) {$\mathcal{DP}(v)$};
          \node at (4,3.5) {$\mathcal{P}(v)$};
          \node at (4,0) {$\mathcal{P}(e)$};

\end{tikzpicture}

The following definitions simply remove the extra numerical values included in the stalks of $\mathcal{DP}$ to obtain values in the stalks of $\mathcal{P}$. For the source and sink, define $$\varphi_{v_S}(e,0) = e$$  $$\varphi_{v_T}(e,x) = e$$  

For any other node $v$, define  $$\varphi_v(e_i,e_j,x) = \{e_i,e_j\}$$  $$\varphi_v(\bot) = \bot$$

And for any edge $e$ and any $x \in [0, \infty]$, define $$\varphi_e(x) = \top$$  $$\varphi_e(\bot) = \bot$$

It can be checked that these choices make the above diagram commute whenever $v \in e$, and therefore give a sheaf morphism $\varphi$.  The image under $\varphi$ of any global section of $\mathcal{DP}$ is a global section of $\mathcal{P}$.  Thus, applying Theorem \ref{thm:SecToPath}, we see that for any global section $s$ of $\mathcal{DP}$, there exists a path from source to sink such that $s(e) \neq \bot$ for all edges $e$ in the path.

Similarly, the image under $\varphi$ of any local section of $\mathcal{DP}$ is a local section of $\mathcal{P}$ defined on the same domain.  In this way, the search over local sections of $\mathcal{DP}$ that takes place in our rewriting of Dijkstra's algorithm can be mapped, through $\varphi$, to a search over local sections of $\mathcal{P}$, as described in section 2.1.

\subsection{Other Algorithms and a General Claim}

This interpretation of paths in the language of sheaves is of course not limited to Dijkstra's algorithm. Another closely related application is the A* algorithm, which is similar to Dijkstra's algorithm, but incorporates a heuristic that estimates the distance from a given node to the sink. The description of Dijkstra's algorithm in the language of $\mathcal{P}$ can be modified to give the A* algorithm by simply changing the cost function to incorporate the heuristic. Similarly, a modified version of $\mathcal{DP}$ could also be used to reinterpret A*. The restriction maps of $\mathcal{DP}$ would need to be adjusted to take into account the heuristic at each node.

Having considered these particular cases, we will now make a general claim about graph algorithms in terms of $\mathcal{P}$. As shown in Theorem \ref{thm:PathToSec}, a path that does not contain the sink or source except possibly at its endpoints corresponds to a local section of $\mathcal{P}$. Thus, if an algorithm for finding a particular path from source to sink searches through the set of paths described in Theorem \ref{thm:PathToSec}, then it can be reinterpreted as searching through the sections of $\mathcal{P}$.  If the search works by extending known paths, then this can be understood as extending local sections, with the goal of finding a particular global section. And finally, any comparison of current paths in the algorithm corresponds to a comparison of current sections, which may be facilitated by some function on the sections, such as the cost function used for Dijkstra's algorithm.

We'll also suggest that $\mathcal{DP}$ and the morphism $\varphi: \mathcal{DP} \rightarrow \mathcal{P}$ can serve as a model for building sheaves specific to a given algorithm. For an algorithm that can be interpreted in the language of $\mathcal{P}$, it may be possible to build a sheaf $\mathcal{F}$ that represents paths and has some of the decision making of the algorithm built in.  Constructing $\mathcal{F}$ in the right way will then allow for a morphism to $\mathcal{P}$ that forgets the additional structure of $\mathcal{F}$.

\section{Suggestions for Future Work}

\begin{enumerate}
  \item \textit{Reinterpreting Other Concepts in Graph Theory}:
  \begin{itemize}
    \item \textit{Cycles.} While $\mathcal{P}$ concerns paths from source to sink, modifying the definition to use graphs without a distinguished source or sink would result in a sheaf concerned with cycles.  The detection of cycles of graphs is key to many applications from material science to space communication networks and a sheaf theoretic description of them may open the door to new tools from topology and category theory for these problems. 
    
    \item \textit{Arbitrary Graph Paths.}
    In our path sheaf, we require a declared source and sink node. However, it would be a useful tool to have a sheaf that could detect path objects as global sections with no designated source or sink. This type of object could prove important for developing a strong language for graph theoretic ideas in the language of cellular sheaves.
    
  \item \textit{Multi-agent Pathfinding.}
    Another problem of interest in graph theory is that of finding paths with multiple sinks and sources and some form of interference between paths. The specifics of the type of interference and exactly what is being optimized change with the specific problem, but the goal is one that is useful to many real world applications. The sheaves that we have described in this paper may be able to be adapted to find solutions of multi-agent pathfinding problems, which are know to be unable to find optimal solutions in polynomial time. Perhaps these sheaves can give us a better understanding of where this process breaks down. 
    
  \item \textit{Maximizing flow.} If we were to consider a different graph for the max flow problem where edge weights represent capacity instead of distance, we can use $\mathcal{P}$ to describe paths from source to sink and then use the capacities as the number of copies of that edge that are allow to be considered in a collection of paths. As described here this approach lacks formality.
  
  \item \textit{Other Applications.}
    While this paper is largely concerned with networking problems, there is a potential to apply these ideas to other graph theoretic concepts.  Sheaves may be an effective tool for describing other problems or objects in graph theory, such as vertex or edge coloring problems or counting the number of spanning trees.  An effective dictionary that could translate between graph theory problems and sheaf theory problems would be of interest to us.
    
  \end{itemize}

  \item \textit{Encoding Algorithms.} Since $\mathcal{P}$ does a good job representing paths over graphs, it may be possible to reinterpret other algorithms using sheaves, or using $\mathcal{P}$ specifically.  In general, there may be good ways to reframe other search algorithms in terms of sheaves.  In particular, we wonder how many algorithms can be described as extending local sections of specific sheaves to global sections.
   \item \textit{Examine existing algorithms for slightly different settings.}
Some path finding algorithms consider a source but no specific sink and search for minimal distances to every node, including Dijkstra's algorithm.  Some also allow for negative edge weights. There are likely sheaves that reflect algorithms, like Bellman-Ford \cite{ref_bellman-ford} or Floyd-Warshall \cite{ref_floyd-warshall}, with different goals or different initial assumptions about the graph.

    \item \textit{Sheaf Characteristics.} There is a wide range of techniques available for analyzing sheaves.  Two major tools are sheaf cohomology and consistency radius.  Sheaf cohomology makes use of algebraic characteristics of the data in the sheaf to detect properties of the global sections.  The sheaves constructed here could be adapted to more easily compute sheaf cohomology, which may lead to interesting interpretations of the sections.
    
    In addition, it may be interesting to see these algorithms adapted and interpreted using consistency radius.  Consistency radius gives a way of testing how close data contained across a sheaf is to being a section (\cite{robinson2017sheaves}, \cite{robinson2018assignments}).  This tool may be useful in different graph theory applications, but it also may be useful in practical situations.  Being able to compare data collected at a node with the expected values given by a sheaf can help with diagnosing problems with network structure or connectivity.
\end{enumerate}

\bibliography{sheaves}
\bibliographystyle{alpha}
\end{document}